\documentclass[10pt,twocolumn,journal]{IEEEtran}
\usepackage{epsf,psfrag,amssymb,amsfonts,graphicx,amsmath,cite}
\usepackage[mathscr]{eucal}

\def\boxit#1{\vbox{\hrule\hbox{\vrule\kern3pt
        \vbox{\kern3pt#1\kern3pt}\kern3pt\vrule}\hrule}}

\def\reals{ { {\rm  I \kern-0.15em R }  } }
\def\complex{ {\,{{\rm C} \kern-0.50em \raise0.20ex {  |}}\, }}

\def\sbf{{\bf s}}

\def\Rbf{{\bf R}}
\def\Sbf{{\bf S}}

\def\Kc{{\cal K}}

\def\Tc{{\cal T}}

\def\defeq{{\stackrel{\Delta}{=}}}
\def\Rxx{\Rbf_{\ssstyle X\kern-.1em X}}

\let\ssstyle=\scriptscriptstyle

\def\ie{{\it i.e.,\ \/}}
\def\Kout{\setbox1=\hbox{\Huge\bf K}\hbox to
1.05\wd1{\hspace{.05\wd1}
}}
\def\Sout{\setbox1=\hbox{\Huge\bf S}\hbox to 1.05\wd1{\hspace{.05\wd1}
}}

\def\ie{{\it i.e.,\ \/}}

\def\defeq{{\stackrel{\Delta}{=}}}

\def\scalefig#1{\epsfxsize #1\textwidth}
\def\nn{{\nonumber}}
\newcommand{\mbbE}{\mathbb{E}}
\newtheorem{lemma}{Lemma}
\newtheorem{theorem}{Theorem}

\newtheorem{corollary}{Corollary}

\begin{document}
\title{\bf \LARGE On Myopic Sensing for Multi-Channel Opportunistic Access:\\[0.2em]
Structure, Optimality, and Performance\thanks{Manuscript received November 30, 2007; revised June 1, 2008 and
June 26, 2008; accepted July 9, 2008. Part of this work
was presented at CogNet, June 2007 and ICASSP, March 2008. This work was supported by the
Army Research Laboratory CTA on Communication and Networks under
Grant DAAD19-01-2-0011 and by the National Science Foundation under
Grants CNS-0627090, ECS-0622200, and CNS-0347621.}}
\author{Qing Zhao, Bhaskar Krishnamachari, Keqin Liu\thanks{Q. Zhao and K. Liu are with the
Department of Electrical and Computer Engineering, University of
California, Davis, CA 95616. Emails:
\{qzhao,kqliu\}@ucdavis.edu. B. Krishnamachari is with the Ming Hsieh Department of Electrical Engineering,
University of Southern California, Los Angeles, CA 90089. Email: bkrishna@usc.edu.}}


\maketitle%
\thispagestyle{empty}

\vspace{-4em}

\begin{abstract}

We consider a multi-channel opportunistic communication system where the states of these channels evolve as
independent and statistically identical Markov chains (the Gilbert-Elliot channel model). A user chooses
one channel to sense and access in each slot and collects a reward determined by the state of the chosen
channel. The problem is to design a sensing policy for channel selection to maximize the average reward,
which can be formulated as a multi-arm restless bandit process.
In this paper, we study the structure, optimality, and performance of the myopic sensing policy.
We show that the myopic sensing policy has a simple robust structure that reduces channel selection to a
round-robin procedure and obviates the need for knowing the channel transition probabilities. The optimality
of this simple policy is established for the two-channel case and conjectured for the general case based
on numerical results. The performance of the myopic sensing policy is analyzed, which, based on the
optimality of myopic sensing, characterizes the maximum throughput of a multi-channel opportunistic communication
system and its scaling behavior with respect to the number of channels. These results apply to cognitive radio networks,
opportunistic transmission in fading environments, downlink scheduling in
centralized networks, and resource-constrained jamming and anti-jamming.

\vspace{0.5em}

\end{abstract}

\section{Introduction}

\subsection{Multi-Channel Opportunistic Access}

The fundamental idea of opportunistic access is to adapt the transmission
parameters (such as data rate and transmission power) according to the state
of the communication environment including, for example, fading conditions, interference
level, and buffer state. Since the seminal work by Knopp and Humblet in 1995 \cite{Knopp&Humblet:95ICC},
the concept of opportunistic access has found applications beyond transmission and scheduling over fading channels.
An emerging application is cognitive radio for opportunistic spectrum access, where secondary users search in the spectrum for
idle channels temporarily unused by primary users \cite{Zhao&Sadler:07SPM}. Another application is resource-constrained jamming
and anti-jamming, where a jammer seeks channels occupied by users or a user tries to avoid jammers.

We consider a general opportunistic communication system where a user has access to $N$ parallel channels
and chooses one channel to sense and access in each slot, aiming to maximize its expected long-term reward (\ie throughput).
This user can be a base station, and each channel is associated with
a downlink receiver. In this case, channel selection is equivalent to receiver selection, and the general problem considered here
also applies to downlink scheduling in a centralized network.

These $N$ channels are modelled as independent and stochastically identical Gilbert-Elliot channels
\cite{Gilbert:60}, which has been commonly used to abstract physical channels with memory
(see, for example, \cite{Zorzi&etal:98COM,Johnston&Krishnamurthy:06TWC}). As illustrated in Fig.~\ref{fig:MC},
the state of a channel --- good or bad ---
indicates the desirability of accessing this channel and determines the resulting reward.
For example, for the application of cognitive radio networks, the good state represents
an unused channel by primary users while the bad state an occupied channel\footnote{When the primary network employs
load balancing across channels, the occupancy process of all channels can be considered stochastically identical.}.
The transitions between these two states follow a Markov chain with transition probabilities $\{p_{ij}\}_{i,j=0,1}$.

\begin{figure}[htb]
\centerline{
\begin{psfrags}
\scalefig{0.9}
\psfrag{A}[c]{ $0$}
\psfrag{B}[c]{ $1$}
\psfrag{A1}[c]{ (bad)}
\psfrag{B1}[c]{ (good)}
\psfrag{a}[c]{$p_{01}$}
\psfrag{b}[l]{ $p_{11}$}
\psfrag{a1}[r]{$p_{00}$}
\psfrag{b1}[c]{ $p_{10}$}
\scalefig{0.4}\epsfbox{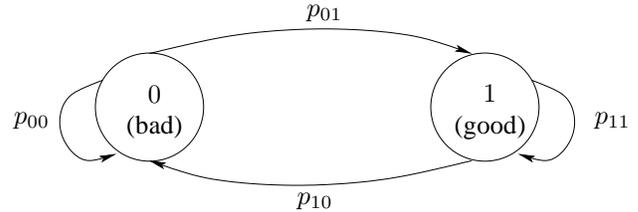}
\end{psfrags}}
\caption{The Gilbert-Elliot channel model.}
\label{fig:MC}
\end{figure}

A sensing policy that governs the channel selection in each slot is crucial
to the efficiency of multi-channel opportunistic access.
The design of the optimal sensing policy can
be formulated as a partially observable Markov decision process
(POMDP) for generally correlated channels, or a
restless multi-armed bandit process for independent channels. Unfortunately, obtaining the optimal
policy for a general POMDP or restless bandit process is often intractable
due to the exponential computation complexity.

A common approach of trading performance for tractable solutions is to consider myopic policies. A myopic policy aims
solely at maximizing the immediate reward, ignoring the impact of the current action on the future reward. Obtaining
a myopic policy is thus a static optimization problem instead of a sequential decision-making problem. As a consequence,
the complexity is significantly reduced, often at the price of considerable performance loss.

In this paper, we show that for designing sensing strategies for multi-channel opportunistic access, low complexity does not necessarily
imply suboptimal performance. The myopic sensing policy with a simple and robust structure achieves the optimal
performance under the i.i.d. Gilbert-Elliot channel model.

\subsection{Contribution}

Under the i.i.d. Gilbert-Elliot channel model,
we establish the structure and optimality of the myopic sensing policy and
analyze its performance.

\subsubsection{Structure of Myopic Sensing}

The first contribution of this paper is the establishment
of a simple and robust structure
of the myopic sensing policy. Besides significant implications in the practical
implementation, this result serves as the key to the optimality
proof and the performance analysis.

We show that the basic structure of the myopic policy
is a round-robin scheme based on a circular ordering of the channels. For the case of
$p_{11}\ge p_{01}$, the circular order is constant and determined by the initial information (if any)
on the state of each channel. The myopic action is to stay in the same channel when
it is good (state $1$) and switch to the next channel in the circular order when
it is bad. In the case of $p_{11}< p_{01}$, the circular order is reversed in every slot
with the initial order determined by the initial information on channel states.
The myopic policy stays in the same channel
when it is bad; otherwise, it switches to the next channel in the current circular
order\footnote{It is easy to show that $p_{11} > p_{01}$ corresponds to the case where
the channel states in two consecutive slots are positively correlated, \ie for any distribution
of $S(t)$, we have $\mbbE[(S(t)-\mbbE[S(t)])(S(t+1)-\mbbE[S(t+1)])]>0$, where $S(t)$ is
the state of the Gilbert-Elliot channel in slot $t$. Similar, $p_{11} < p_{01}$ corresponds to the case where
$S(t)$ and $S(t+1)$ are negatively correlated, and $p_{11}=p_{01}$ the case where
$S(t)$ and $S(t+1)$ are independent.}.

The significance of this result in terms of the practical implementations
of myopic sensing is twofold. First, it demonstrates the simplicity of
myopic sensing: channel selection is reduced to a simple round-robin procedure.
The myopic sensing policy requires no computation and little memory.
Second, it shows that myopic sensing is robust to model mismatch. Specifically,
the myopic sensing policy has a semi-universal structure; it can be implemented without
knowing the channel transition probabilities. The only required information
about the channel model is
the order of $p_{11}$ and $p_{01}$. As a result, the myopic sensing policy
automatically tracks variations in the channel model provided that the order of
$p_{11}$ and $p_{01}$ remains unchanged. Note that when $p_{11}=p_{01}$, channel
states become independent in time; all channel selections lead to the same
performance. We thus expect that myopic sensing is robust to estimation errors
in the order of $p_{11}$ and $p_{01}$, which usually occur when
$p_{11}\approx p_{01}$. This has been confirmed by simulation results~\cite{Liu&Zhao:07TR}.

\subsubsection{Optimality of Myopic Sensing}

Surprisingly, the myopic sensing policy with such a simple and
robust structure is, in fact, optimal as established in this paper
for $N=2$. Based on numerical results, we conjecture that the
optimality of the myopic policy can be generalized to $N>2$. The
optimality along with the simple and robust structure makes the
myopic sensing policy particularly appealing.

In a recent work \cite{Javidi&etal:08ICC}, based on the structure of the myopic policy,
the optimality result has been extended to $N>2$ under the condition of
$p_{11}\ge p_{01}$. While numerical results indicate that for a wide range of $p_{11}$ and $p_{01}$, the
myopic policy is also optimal for $N>2$ with $p_{11}< p_{01}$, pathological cases where optimality fails have been found
when $p_{01}-p_{11}$ is close to $1$. Nevertheless, the performance loss of the myopic policy in these cases
is minimal and tends to diminish with the horizon length. Establishing necessary and/or sufficient conditions
(potentially in the form of bounding $p_{01}-p_{11}$)
under which the myopic policy is optimal for $p_{11}< p_{01}$ appears to be
challenging. It is our hope that results and approaches presented in this paper, in particular, the simple
structure of the myopic policy, may
stimulate fresh ideas for completing the picture on the optimality of the myopic policy.

\subsubsection{Performance of Myopic Sensing}

The optimality of the myopic sensing policy motivates the performance analysis, as its performance
defines the throughput limit of a multi-channel opportunistic communication
system under the i.i.d. Gilbert-Elliot channel model. We are particularly
interested in the relationship between the maximum throughput and the number of channels.

Closed-form expressions for the performance of POMDP and restless bandit policies are rare. For this problem at hand,
the simple structure of the myopic policy again renders an exception.
Specifically, based on the structure of the myopic policy,
we show that its performance
is determined by the stationary distributions of a higher-order
countable-state Markov chain.
For $N=2$, we have a first-order Markov chain whose stationary distribution can be obtained in closed-form, leading
to exact characterizations of the throughput. For $N>2$, we construct first-order Markov
processes that stochastically dominate or are dominated by this higher-order Markov chain.
The stationary distributions of the former, again obtained in closed-forms, lead to lower and
upper bounds that monotonically tighten as the number $N$
of channels increases.

These analytical characterizations allow us to study
the rate at which the maximum throughput of an opportunistic system increases with $N$,
and to obtain the limiting performance as $N$ approaches to infinity.
Our result demonstrates that the maximum throughput of
a multi-channel opportunistic system with single-channel sensing saturates at geometric rate as the number of
channels increases. This result suggests to system designers the importance of having radios capable of sensing multiple
channels in order to fully exploit the communication opportunities offered by
a large number of channels.

\begin{figure*}[t!]
\normalsize
\begin{equation}
\omega_i(t+1)=\left\{\begin{array}{ll}
p_{11}, &  a(t)=i, S_{a(t)}(t)=1\\
p_{01}, &  a(t)=i, S_{a(t)}(t)=0\\
\omega_i(t)p_{11} + (1-\omega_i(t))p_{01}, & a(t)\neq i\\
\end{array}
\right..
\label{eq:omega}
\end{equation}
\hrulefill
\end{figure*}

\subsection{Related Work}

The structure, optimality, and performance analysis of myopic sensing in the context of opportunistic
access may bear significance in the general context of restless multi-armed
bandit processes. While an index policy (Gittins index \cite{gittins}) is known to be optimal for
the classical bandit problems, the structure of the optimal policy for a general restless bandit process
remains unknown, and the problem is shown to be PSPACE-hard \cite{tsitsiklis}. Whittle
proposed a Gittins-like heuristic index policy for restless bandit
problems~\cite{whittle}, which is asymptotically optimal
in certain limiting regime \cite{weber}. Beyond this asymptotic
result, relatively little is known about the structure of the
optimal policies for a general restless bandit process. The existing literature mainly focuses on
approximation algorithms and heuristic policies~\cite{Bertsimas00,Guha07}. The optimality of the myopic policy
shown in this paper suggests non-asymptotic conditions under which an index policy
can be optimal for restless bandit processes.

The results presented in this paper apply to cognitive radio networks, which has received
increasing attention recently. In this context,
the design of sensing policies for tracking the rapidly varying spectrum opportunities
has been addressed in \cite{Zhao&etal:07JSAC,Chen&etal:08IT} under a general
Markvian model of potentially correlated channels, where a POMDP framework has been
developed.

This paper is also related to channel probing and transmission strategies in multichannel
wireless systems (see \cite{Sabharwal&etal:07ToN,Guha&etal:06CISS,Chang&Liu:07MobiCom,Agarwal&Honig:07CrownCom}
and references therein). In contrast to the Markovian model considered in this paper, these existing results
adopt a memoryless channel model.

\section{Problem Formulation}
\label{sec:formulation}

We consider the scenario where a user is trying to access $N$ independent and stochastically identical channels
using a slotted transmission structure.
The state $S_i(t)$ of channel $i$ in slot $t$ is given by a two-state Markov
chain shown in Fig.~\ref{fig:MC}.
At the beginning of each slot, the user selects one of the $N$
channels to sense. If the channel is sensed to be good
(state $1$), the user transmits and collects one unit of
reward.  Otherwise, the user does not transmit (or transmits at a
lower rate), collects no reward, and waits until the next slot to
make another choice. The objective is to maximize the average reward
(throughput) over a horizon of $T$ slots by choosing judiciously a
sensing policy that governs channel selection in each slot.

Due to limited sensing, the full system state $[S_1(t),\cdots,S_N(t)]\in \{0,1\}^N$ in slot $t$
is not observable.  The user, however, can infer the state from its decision and
observation history. It has been
shown that a sufficient statistic for optimal decision making
is given by the conditional probability that each channel
is in state $1$ given all past decisions and observations \cite{Smallwood&Sondik:71OR}.
Referred to as the belief vector, this sufficient statistic is denoted by
$\Omega(t) \,\defeq\, [\omega_1(t),\cdots,\omega_N(t)]$, where $\omega_i(t)$
is the conditional probability that $S_i(t)=1$.
Given the sensing action $a(t)$ and the observation $S_{a(t)}(t)$ in slot $t$, the
belief vector for slot~$t+1$ can be obtained via Bayes Rule as given in~\eqref{eq:omega}.

A sensing policy $\pi$ specifies a sequence of functions $\pi =
[\pi_1, \pi_2, \cdots, \pi_T]$, where $\pi_t$ is the decision rule at time $t$ that maps a belief vector
$\Omega(t)$ to a sensing action $a(t)\in\{1,\cdots,N\}$ for slot
$t$. Multi-channel opportunistic access can thus be formulated as
the following stochastic control problem.
\begin{equation}
\pi^*=\arg\max_\pi \mbbE_{\pi}\left[\sum_{t=1}^T R_{\pi_t(\Omega(t))}(t)|\Omega(1)\right],
\label{eq:pi*}
\end{equation}
where $\pi_t(\Omega(t))$ is the channel selected and $R_{\pi_t(\Omega(t))}(t)=S_{\pi_t(\Omega(t))}(t)$ the reward so obtained
when the belief is $\Omega(t)$, and $\Omega(1)$ is the initial belief vector.
If no information about the initial system state is available, each entry of $\Omega(1)$ can be set to the
stationary distribution $\omega_o$ of the underlying Markov chain:
\begin{equation}
\omega_o=\frac{p_{01}}{p_{01}+p_{10}}.
\label{eq:wo}
\end{equation}

This problem falls into
the general model of POMDP. It can also be considered as a restless
multi-armed bandit problem by treating the belief value of each channel
as the state of each arm of a bandit.
Note that for a given sensing policy $\pi$, the belief vectors $\{\Omega(t)\}_{t=1}^T$ form a Markov
process with an uncountable state space. The expectation in \eqref{eq:pi*} is with respect to this
Markov process which determines the reward process. The difficulty in obtaining the optimal policy
$\pi^*$ and characterizing its performance largely results from the complexity of analyzing a Markov process
with uncountable state space.

\section{Optimal Policy vs. Myopic Policy}

\subsection{Value Function and Optimal Policy}

Let $V_t(\Omega(t))$ be the value function, which represents the
maximum expected total reward that can be obtained starting from
slot $t$ given the current belief vector $\Omega(t)$. Given that the
user takes action $a$ and observes $S_a(t)$ in slot $t$, the reward that can be
accumulated starting from slot $t$ consists of two parts: the expected
immediate reward $\mbbE[R_a(t)]=\mbbE[S_a(t)]=\omega_a(t)$ and the maximum expected
future reward $V_{t+1}(\Tc(\Omega(t)|a, S_a(t)))$, where $\Tc(\Omega(t)|a,
S_a(t))$ denotes the updated belief vector for slot $t+1$ as given in
\eqref{eq:omega}. Averaging over all possible observations $S_a(t)$ and
maximizing over all actions $a$, we arrive at the following
optimality equations.
\begin{eqnarray}
V_{T}(\Omega(T)) &=& \max_{a=1, \cdots, N} \omega_a(T) \nonumber\\
V_{t}(\Omega(t)) &=& \max_{a=1, \cdots, N} \left\{ \omega_a(t) +
\omega_a(t) V_{t+1} \left( \Tc\left( \Omega(t)| a, 1\right)\right) \right.\nn\\
& & +\left.
(1-\omega_a(t))  V_{t+1} \left( \Tc\left( \Omega(t) | a, 0\right) \right)\right\}.
\label{DP-finite-t}
\end{eqnarray}

In theory, the optimal policy $\pi^*$ and its performance
$V_1(\Omega(1))$ can be obtained by solving the above dynamic
program. Unfortunately, this approach is computationally prohibitive
due to the impact of the current action on the future reward and the
uncountable space of the belief vector $\Omega(t)$. Even if
approximate numerical solutions are feasible, they do not provide
insights for system design or analytical characterizations of the
optimal performance $V_1(\Omega(1))$.

\subsection{Myopic Policy}

A myopic policy ignores the impact of the current action on the future reward,
focusing solely on maximizing the expected immediate reward $\mbbE[R_a(t)]$. Myopic policies
are thus stationary: the mapping from
belief vectors to actions does not change with time $t$.  The myopic action $\hat{a}(t)$
and the value function $\hat{V}_t(\Omega(t))$ of the myopic policy for a given belief vector $\Omega(t)$ are given by
\begin{eqnarray}
\hat{a}(t)&=&\arg\max_{a=1,\cdots,N} \omega_a(t),\label{eq:a*}\\
\hat{V}_t(\Omega(t))&=&\omega_{\hat{a}(t)}(t)+\omega_{\hat{a}(t)}(t) \hat{V}_{t+1} \left( \Tc\left( \Omega(t)| \hat{a}(t), 1\right)\right)\nn\\
& &   +
(1-\omega_{\hat{a}(t)}(t))  \hat{V}_{t+1} \left( \Tc\left( \Omega(t) | \hat{a}(t), 0\right) \right).\nn
\end{eqnarray}
In general, obtaining the myopic action in each slot requires
the recursive update of the belief vector $\Omega(t)$ as given in
\eqref{eq:omega}, which requires the knowledge of the transition
probabilities $\{p_{ij}\}$. In the next section, we show that the myopic
policy has a simple semi-universal structure that does not need the update of the
belief vector or the knowledge of the transition
probabilities.

\section{Structure of Myopic Sensing}
\label{sec:structure}

In this section, we establish the simple and robust structure of the myopic policy, which
lays out the foundation for the optimality proof and performance
analysis in subsequent sections.

\subsection{Structure}

The basic element in the structure of the myopic policy is a circular ordering $\Kc$ of the channels.
For a circular order, the starting point is irrelevant: a circular order
$\Kc=(n_1,n_2,\cdots,n_N)$ is equivalent to
$(n_i,n_{i+1},\cdots, n_N, n_1, n_2,\cdots, n_{i-1})$ for any $1\le i\le N$.
An example of a circular order is given in Fig.~\ref{fig:structure1}, where all $N$ channels are
placed on a circle in the clockwise direction.

We now introduce the following notations. For a circular order $\Kc$, let $-\Kc$ denote its reverse
circular order, \ie  for $\Kc=(n_1,n_2,\cdots,n_N)$,
we have $-\Kc=(n_N,n_{N-1},\cdots,n_1)$
(see Fig.~\ref{fig:structure2} for an illustration where the lower circle on the right shows
the reverse circular order of that given by the circle on the left).

For a channel~$i$, let $i_{\Kc}^+$ denote the next channel in the circular order $\Kc$. For example,
for $\Kc=(1,2,\cdots,N)$, we have $i_{\Kc}^+=i+1$ for $1\le i<N$ and $N_{\Kc}^+=1$.

With these notations, we present the structure of the myopic policy in Theorem~\ref{thm:structure}.

\begin{theorem} {\it Structure of Myopic Sensing:}\\
Let $\Omega(1)=[\omega_1(1),\cdots,\omega_N(1)]$ denote the initial belief vector.
The circular channel order $\Kc(1)$ in slot~$1$ is
determined by a descending order of $\Omega(1)$ (\ie $\Kc(1)=(n_1,n_2,\cdots,n_N)$ implies that
$\omega_{n_1}(1)\ge\omega_{n_2}(1)\ge\cdots\ge\omega_{n_N}(1)$).
Let $\hat{a}(1)=\arg\max_{i=1,\cdots,N} \omega_i(1)$. The myopic action $\hat{a}(t)$ in slot $t$ ($t>1$) is given as follows.
\begin{itemize}
\item {\it Case 1: $p_{11}\ge p_{01}$}
\end{itemize}
\begin{equation}
\hat{a}(t)=\left\{\begin{array}{ll}
\hat{a}(t-1), & \mbox{if } S_{\hat{a}(t-1)}(t-1)=1\\
\hat{a}(t-1)_{\Kc(t)}^+, & \mbox{if } S_{\hat{a}(t-1)}(t-1)=0 \\
\end{array}
\right.,
\label{eq:structure1}
\end{equation}
where $\Kc(t)=\Kc(1)$.
\begin{itemize}
\item {\it Case 2: $p_{11}< p_{01}$}
\end{itemize}
\begin{equation}
\hat{a}(t)=\left\{\begin{array}{ll}
\hat{a}(t-1) & \mbox{if } S_{\hat{a}(t-1)}(t-1)=0\\
\hat{a}(t-1)_{\Kc(t)}^+ & \mbox{if } S_{\hat{a}(t-1)}(t-1)=1 \\
\end{array}
\right.,
\label{eq:structure2}
\end{equation}
where $\Kc(t)=\Kc(1)$ when $t$ is odd and $\Kc(t)=-\Kc(1)$ when $t$ is even.
\label{thm:structure}
\end{theorem}


\begin{proof}
See Appendix A.
\end{proof}

Theorem~\ref{thm:structure} shows that the basic structure of the myopic policy
is a round-robin scheme based on a circular ordering of the channels. For
$p_{11}\ge p_{01}$, the circular order is constant: $\Kc(t)=\Kc(1)$ in every slot~$t$, where $\Kc(1)$ is determined by a descending
order of the initial belief values. The myopic action is to stay in the same channel when
it is good (state $1$) and switch to the next channel in the circular order when
it is bad (see Fig.~\ref{fig:structure1} for an illustration).

\begin{figure}[h]
\centerline{
\begin{psfrags}
\psfrag{s}[c]{\small $S_1=0$}
\psfrag{t}[c]{\small $\Kc(t)=\Kc(1)$}
\scalefig{0.3}\epsfbox{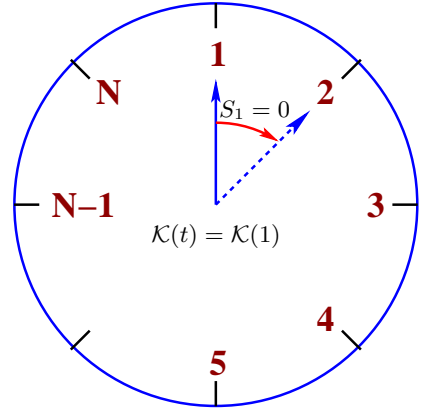}
\end{psfrags}}
\caption{The structure of the myopic policy for $p_{11}\ge p_{01}$: the circular order of the channels is constant and determined
by the initial belief $\Omega(1)$ ($\omega_1(1)\ge \omega_2(1)\ge \cdots \ge \omega_N(1)$ is assumed in this example, thus
$\hat{a}(1)=1$);
the myopic policy switches to the next channel when the current one is in the bad state.}
\label{fig:structure1}
\end{figure}

In the case of $p_{11}< p_{01}$, the circular order is reversed in every slot, \ie
$\Kc(t)=\Kc(1)$ when $t$ is odd and $\Kc(t)=-\Kc(1)$ when $t$ is even, where
the initial order $\Kc(1)$ is determined by the initial belief values. The myopic policy stays in the same channel
when it is bad; otherwise, it switches to the next channel in the {\it current} circular order $\Kc(t)$, which is
either $\Kc(1)$ or $-\Kc(1)$ depending on whether the current time $t$ is odd or even.
An illustrated is given in Fig.~\ref{fig:structure2}.

\begin{figure}[h]
\centerline{
\begin{psfrags}
\psfrag{a}[c]{ $t=1$}
\psfrag{a1}[c]{\small $\Kc(1)$}
\psfrag{b}[c]{ $t=L$}
\psfrag{b1}[c]{\small $\Kc(L)=\Kc(1)$}
\psfrag{c}[c]{ $t=L$}
\psfrag{c1}[c]{\small $~~\Kc(L)=-\Kc(1)$}
\psfrag{o}[c]{ $L$ odd}
\psfrag{e}[c]{ $L$ even}
\psfrag{s}[c]{\footnotesize $S_1=1$}
\scalefig{0.5}\epsfbox{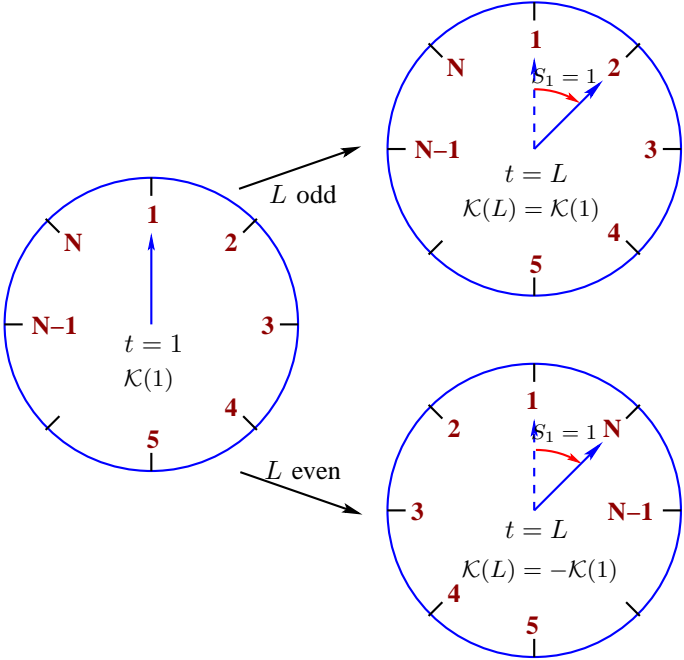}
\end{psfrags}
}
\caption{The structure of the myopic policy for $p_{11}< p_{01}$: in the first slot ($t=1$), the circular order $\Kc(1)$ is determined
by the initial belief $\Omega(1)$ ($\omega_1(1)\ge \omega_2(1)\ge \cdots \ge \omega_N(1)$ is assumed in this example, thus
$\hat{a}(1)=1$). Suppose that channel~$1$ is in the bad state in slots $1,\cdots,L-2$ and in the good state in slot $L-1$. The circular
order at $t=L$ is $\Kc(1)$ when $L$ is odd and $-\Kc(1)$ when $L$ is even, and $\hat{a}(L)$ is the next channel in $\Kc(L)$,
\ie $\hat{a}(L)=2$ for $L$ odd and $\hat{a}(L)=N$ for $L$ even.}
\label{fig:structure2}
\end{figure}

An alternative way to see the channel switching structure of the myopic policy is through the last visit to each channel (once every channel
has been visited at least once). Specifically, for $p_{11}\ge p_{01}$, when
a channel switch is needed, the policy
selects the channel visited the longest time ago. For $p_{11}< p_{01}$,
when a channel switch is
needed, the policy selects, among those channels to which the last
visit occurred an even number of slots ago, the one most recently
visited. If there are no such channels, the user chooses the channel
visited the longest time ago (see Appendix~B for a proof).

\begin{figure*}[t!]
\normalsize
\begin{eqnarray}
p_{11}\ge p_{01}~~~~~~~~~~~~~~~~~~~~ & ~~~~~~ & ~~~~~~~~~~~~~~~~~~~~ p_{11}< p_{01} \nn\\[0.8em]
q_{\vec{i},\, \vec{j}}=\left\{\begin{array}{ll}
\prod_{k=1}^N p_{i_k,j_k} & \mbox{ if } i_1=1\\
p_{i_1,j_N}\prod_{k=2}^{N}p_{i_k,j_{k-1}} & \mbox{ if } i_1=0\\
\end{array}
\right.,
& ~~~~~~ &
q_{\vec{i},\, \vec{j}}=\left\{\begin{array}{ll}
\prod_{k=1}^N p_{i_k,j_{N-k+1}} & \mbox{ if } i_1=1\\
p_{i_1,j_1}\prod_{k=2}^{N}p_{i_k,j_{N-k+2}} & \mbox{ if } i_1=0\\[0.3em]
\end{array}
\right.,
\label{eq:qij}
\end{eqnarray}
where $\vec{i}=[i_1,i_2,\cdots,i_N]$, $\vec{j}=[j_1,j_2,\cdots,j_N]$ with entries equal to $0$ or $1$.\\
\hrule
\end{figure*}

\subsection{Properties}

The simple structure of the myopic policy has significant implications in both practical
and technical aspects. Implementation-wise,
the following properties of the myopic policy follow from its structure: {\it belief-independence} and
{\it model-insensitivity}. Specifically, the myopic policy does not require the update of
the belief vectors or the knowledge
of the transition probabilities except the order of $p_{11}$ and $p_{01}$.
These properties make the myopic policy particularly
attractive in implementation. Besides its simplicity, this semi-universal structure leads
to robustness against model mismatch and variations.

A technical benefit of this simple structure is that it provides the foundation for establishing
the optimality and characterizing the performance of the myopic policy as given in
Sec.~\ref{sec:optimality}-\ref{sec:performance}, as well as the generalizations of the optimality proof
to $N>2$ given in \cite{Javidi&etal:08ICC}. The reason is that the structure allows us to work with a Markov reward
process with a finite state space instead of one with an uncountable state space (\ie belief vectors) as we
encounter in a general POMDP. Details are stated in the corollary below.

\begin{corollary}
Let $\Kc(t)=(n_1,n_2,\cdots,n_N)$ ($n_i\in\{1,2,\cdots,N\}~\forall i$) be the circular order of channels
in slot $t$, where
the starting point of the circular order is fixed to the myopic action: $n_1=\hat{a}(t)$ for all $t$. Then the resulting ordered
channel states
$\vec{\Sbf}(t)\, \defeq\, [S_{n_1}(t),S_{n_2}(t),\cdots,S_{n_N}(t)]\}$ form a $2^N$-state Markov chain with transition
probabilities $\{q_{\vec{i},\, \vec{j}}\}$ given in \eqref{eq:qij}, and
the performance of the myopic policy is determined by the Markov reward process $(\vec{\Sbf}(t),R(t))$ with
$R(t)=S_{n_1}(t)$.
\label{thm:MRP}
\end{corollary}

\begin{proof} The proof follows directly from Theorem~\ref{thm:structure} by noticing that $S_{n_1}(t)$ determines
the channel ordering in $\vec{\Sbf}(t+1)$ and each channel evolves as independent Markov chains. Specifically, for $p_{11}\ge p_{01}$,
if $S_{n_1}(t)=1$, the channel ordering in $\vec{\Sbf}(t+1)$ is the same as that in $\vec{\Sbf}(t)$; if $S_{n_1}(t)=0$,
the first channel (channel $n_1$) in $\vec{\Sbf}(t)$ is moved to the last one in $\vec{\Sbf}(t+1)$ with the ordering of the rest $N-1$ channels
unchanged. For $p_{11}< p_{01}$, if $S_{n_1}(t)=0$, the first channel in $\vec{\Sbf}(t)$ remains the first in $\vec{\Sbf}(t+1)$
while the ordering of the rest channels is reversed; if $S_{n_1}(t)=1$, the ordering of all $N$ channels are reversed.
The transition probabilities given in \eqref{eq:qij} thus follow.
\end{proof}

\section{Optimality of Myopic Sensing}
\label{sec:optimality}

In this section, we establish the optimality of the myopic policy for $N=2$.
Our proof hinges on the structure of the myopic policy given in Theorem~\ref{thm:structure}
and Corollary~\ref{thm:MRP}.

\begin{theorem} {\it Optimality of Myopic Sensing:}\\
For $N=2$, the myopic sensing policy is optimal, \ie $\hat{V}_t(\Omega)=V_t(\Omega)$ for all $t$ and $\Omega$.
\label{thm:optimal}
\end{theorem}

\begin{proof}
see Appendix~C.
\end{proof}

Based on extensive numerical results, we conjecture that the
optimality of the myopic sensing policy can be generalized to
$N > 2$. A recent work~\cite{Javidi&etal:08ICC} has made partial
progress towards proving this conjecture, by showing that the
optimality holds for $N
> 2$ under the condition of $p_{11}\ge p_{01}$.
Furthermore, it is shown in~\cite{Javidi&etal:08ICC} that if the myopic
policy is optimal under the sum-reward criterion over a finite horizon, it is
also optimal for other criteria such as discounted and averaged rewards over a
finite or infinite horizon. In
the case of infinite-horizon discounted reward, it is determined
that so long as the discount factor is less than 0.5, the myopic
policy is optimal for all $N$.

\section{Performance of Myopic Sensing}
\label{sec:performance}

In this section, we analyze the performance of the myopic policy. With the optimality results,
the throughput achieved
by the myopic policy defines the performance limit of a multi-channel opportunistic communications
system. In particular, we are
interested in the relationship between this maximum throughput
and the number $N$ of channels.

\subsection{Uniqueness of Steady-State Performance and Its Numerical Evaluation}

We first establish the existence and uniqueness of the system steady states under the myopic policy.
The steady-state throughput of the myopic policy is given by
\begin{equation}
U(\Omega(1))\,\defeq\, \lim_{T\rightarrow\infty} \frac{\hat{V}_{1:T}(\Omega(1))}{T},
\label{eq:Udef}
\end{equation}
where $\hat{V}_{1:T}(\Omega(1))$ is the expected total reward obtained in $T$ slots under the myopic policy
when the initial belief is $\Omega(1)$. From Corollary~\ref{thm:MRP}, $U(\Omega(1))$ is determined by the Markov reward
process $\{\vec{\Sbf}(t),R(t)\}$. It is easy to see that the $2^N$-state Markov chain $\{\vec{\Sbf}(t)\}$ is
irreducible and aperiodic, thus has a limiting distribution. As a consequence, the limit in \eqref{eq:Udef} exists,
and the steady-state throughput $U$ is independent of the initial belief value $\Omega(1)$.

Corollary~\ref{thm:MRP} also provides a numerical approach to evaluating $U$ by calculating the limiting (stationary)
distribution of $\{\vec{\Sbf}(t)\}$ whose transition probabilities are given in~\eqref{eq:qij}.
Specifically, the throughput $U$ is given by the summation of the
limiting probabilities of those $2^{N-1}$ states with first entry $S^{(1)}=1$.
This numerical approach, however, does not provide an analytical characterization
of the throughput $U$ in terms of the number $N$ of channels and the transition probabilities $\{p_{i,j}\}$.
In the next section, we obtain analytical expressions of $U$ and its scaling behavior
with respect to $N$ based on a stochastic dominance argument.

\subsection{Analytical Characterization of Throughput}

\subsubsection{The Structure of Transmission Period}
\label{sec:TP}

From the structure of the myopic policy we can see that the key to
the throughput is how often the user switches channels, or
equivalently, how long the user stays in the same channel. When
$p_{11}\ge p_{01}$, the event of channel switching is equivalent to a
slot {\em without} reward. The opposite holds when $p_{11}< p_{01}$:
a channel switching corresponds to a slot {\em with} reward.

We thus introduce the concept of transmission period (TP), which is the time the user stays in
the same channel (see Fig.~\ref{fig:TP}). Let $L_k$ denote the length of the
$k$th TP. We then have a discrete-time random process $\{L_k\}_{k=1}^\infty$
with a state space of positive integers.

\begin{figure}[h]
\centerline{
\begin{psfrags}
\psfrag{c}[c]{ channel switching}
\psfrag{l1}[c]{ $L_k=3$}
\psfrag{l2}[c]{ $L_{k+1}=6$}
\psfrag{t}[c]{ $t$}
\scalefig{0.5}\epsfbox{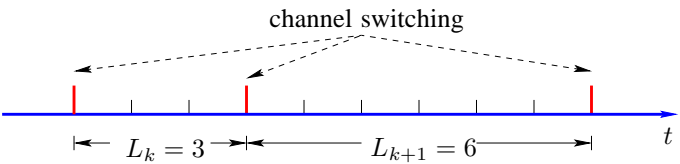}
\end{psfrags}}
\caption{The transmission period structure.}
\label{fig:TP}
\end{figure}

Based on the structure of the myopic policy, we have
\begin{equation}
U=\left\{\begin{array}{ll}
\lim_{K\rightarrow\infty}\frac{\Sigma_{k=1}^{K}(L_k-1)}{\Sigma_{k=1}^{K}L_k}, & p_{11}\ge p_{01}\\
\lim_{K\rightarrow\infty}\frac{\Sigma_{k=1}^{K}1}{\Sigma_{k=1}^{K}L_k}, & p_{11}< p_{01}.\\
\end{array}
\right..
\label{eq:U0}
\end{equation}
Let $\bar{L}=\lim_{K\rightarrow\infty}\frac{\sum_{k=1}^K L_k}{K}$ denote the average length of a TP.
The above equation leads to
\begin{equation}
U=\left\{\begin{array}{ll}
1-1/\bar{L}, & p_{11}\ge p_{01}\\[0.5em]
1/\bar{L}, & p_{11}< p_{01}\\
\end{array}
\right..
\label{eq:U}
\end{equation}
Throughput analysis is thus reduced to analyzing the average TP length $\bar{L}$.
For $N=2$, a closed-form expression
of $\bar{L}$ can be obtained, which leads to a closed-form expression
of the throughput $U$ (see Sec.~\ref{sec:n2}). For $N>2$, lower and upper bounds on $U$ are obtained
(see Sec.~\ref{sec:n3}).

\vspace{1em}

\subsubsection{Throughput for $N=2$}
\label{sec:n2}

From the structure of the myopic policy, $\{L_k\}_{k=1}^\infty$ form a first-order Markov chain for $N=2$.
Specifically, the distribution of $L_k$ is determined by the belief value of the chosen channel in the first
slot of the $k$-th TP. The latter equals to $p_{01}^{(L_{k-1}+1)}$ for $p_{11}\ge p_{01}$
and $p_{11}^{(L_{k-1}+1)}$ for $p_{11}< p_{01}$, where $p_{01}^{(j)}$ is the $j$-step transition probability.
The transition probabilities of $\{L_k\}_{k=1}^\infty$ are thus given as follows.
\begin{itemize}
\item For $p_{11}\ge p_{01}$,
\begin{equation}
r_{ij}=\left\{\begin{array}{ll}
1-p^{(i+1)}_{01}, & i\geq 1,~j=1\\
p^{(i+1)}_{01}p_{11}^{j-2}p_{10}, & i\geq 1,~j\geq 2.\\
\end{array}
\right..
\label{eq:rij1}
\end{equation}
\item For $p_{11}< p_{01}$,
\begin{equation}
r_{ij}=\left\{\begin{array}{ll}
p^{(i+1)}_{11}, & i\geq 1,~j=1\\
p^{(i+1)}_{10}p_{00}^{j-2}p_{01}, & i\geq 1,~j\geq 2\\
\end{array}
\right..
\label{eq:rij2}
\end{equation}
\end{itemize}

As shown in Appendix~D, the limiting distribution $\{\lambda_l\}_{l=1}^\infty$ of this countable-state Markov chain
can be obtained in closed-form, which leads to $\bar{L}=\sum_{l=1}^\infty l\lambda_l$ and then the throughput $U$ from~\eqref{eq:U}.

\begin{theorem}
For $N=2$, the throughput $U$ is given by
\begin{equation}
U=\left\{\begin{array}{ll}
1-\frac{1-p_{11}}{1+\bar{\omega}-p_{11}}, & p_{11}\ge p_{01}\\[0.5em]
\frac{p_{01}}{1-\bar{\omega}'+p_{01}}, & p_{11}< p_{01}\\
\end{array}
\right.,
\label{eq:N2}
\end{equation}
where $\bar{\omega}$ and $\bar{\omega}'$ are the expected probability that the channel
the user switches to is in state $1$ when $p_{11}\ge p_{01}$ and $p_{11}< p_{01}$, respectively. They
are given in~\eqref{eq:baromega1} and \eqref{eq:baromega2}.
\label{thm:N2}
\end{theorem}

\begin{proof}
See Appendix~D.
\end{proof}

\begin{figure*}[t!]
\begin{eqnarray}
\bar{\omega}&=&\frac{p_{01}^{(2)}}{1+p_{01}^{(2)}-A},~~~\mbox{where
}~
p_{01}^{(2)}=p_{00}p_{01}+p_{01}p_{11},~ A=\frac{p_{01}}{1+p_{01}-p_{11}}(1-\frac{(p_{11}-p_{01})^3(1-p_{11})}{1-(p_{11})^2+p_{11}p_{01}}),\label{eq:baromega1}\\
\bar{\omega}'&=&\frac{B}{1-p_{11}^{(2)}+B},~~~\mbox{where}~
p_{11}^{(2)}=p_{10}p_{01}+p_{11}p_{11},~
B=\frac{p_{01}}{1+p_{01}-p_{11}}(1+\frac{(p_{11}-p_{01})^3(1-p_{11})}{1-(1-p_{01})(p_{11}-p_{01})}).\label{eq:baromega2}
\end{eqnarray}
\hrulefill
\end{figure*}

\vspace{1em}

\subsubsection{Throughput for $N>2$}
\label{sec:n3}

For $N>2$, $\{L_k\}_{k=1}^\infty$ is a random process with higher-order memory. In particular,
for $p_{11}\ge p_{01}$, it is an $(N-1)$-th order Markov chain. As a consequence, closed-form expressions
of $\bar{L}$ are difficult to obtain.
Our objective is to develop lower and upper bounds on $U$, which would allow
us to study the scaling behavior of $U$ with respect to $N$.

The approach is to construct first-order Markov chains that
stochastically dominate or are dominated by $\{L_k\}_{k=1}^\infty$.
The stationary distributions of these first-order Markov chains,
which can be obtained in closed-form, lead to lower and upper
bounds on $U$ according to \eqref{eq:U}. Specifically, for
$p_{11}\ge p_{01}$, a lower bound on $U$ is obtained by constructing
a first-order Markov chain whose stationary distribution is
stochastically dominated by the stationary distribution of
$\{L_k\}_{k=1}^\infty$. An upper bound on $U$ is given by a
first-order Markov chain whose stationary distribution
stochastically dominates the stationary distribution of
$\{L_k\}_{k=1}^\infty$. Similarly, bounds on $U$ can be obtained for
$p_{11}< p_{01}$.

\begin{theorem}
For $N>2$, we have the following lower and upper bounds on the throughput $U$.
\begin{itemize}
\item {\it Case 1:} $p_{11}\ge p_{01}$
\end{itemize}
\begin{equation}
\frac{C}{C+(1-D+C)(1-p_{11})}\leq U \leq\frac{\omega_o}{1-p_{11}+\omega_o},
\label{eq:N3a}
\end{equation}
where $\omega_o$ is given by \eqref{eq:wo} and
\begin{eqnarray}
C&=&\omega_o(1-(p_{11}-p_{01})^{N}),\nn\\
D&=&\omega_o(1-\frac{(p_{11}-p_{01})^{N+1}(1-p_{11})}{1-p_{11}^2+p_{11}p_{01}}).\nn
\end{eqnarray}

\vspace{1em}

\begin{itemize}
\item {\it Case 2:} $p_{11}< p_{01}$
\end{itemize}

\begin{equation}
1-\frac{p_{10}^{(2)}}{E-p_{01} H}\leq U \leq 1-
\frac{p_{10}^{(2)}}{E-p_{01} G}, \label{eq:N3b}
\end{equation}
where
\begin{eqnarray}
p_{10}^{(2)}&=&p_{10}p_{00}+p_{11}p_{10},\nn\\
E&=&p_{10}^{(2)}(1+p_{01})+p_{01}(1-F),\nn\\
F&=&(1-p_{01})(1-\omega_o)\nn\\
& &(\frac{1}{2-p_{01}}-\frac{p_{01}(p_{11}-p_{01})^4}{1-(p_{11}-p_{01})^2(1-p_{01})^2}),\nn\\
G&=&(1-\omega_o)(\frac{1}{2-p_{01}}-\frac{p_{01}(p_{11}-p_{01})^6}{1-(p_{11}-p_{01})^2(1-p_{01})^2}),\nn\\
H&=&(1-\omega_o)(\frac{1}{2-p_{01}}-\frac{p_{01}(p_{11}-p_{01})^{2N-1}}{1-(p_{11}-p_{01})^2(1-p_{01})^2}).\nn
\end{eqnarray}

\vspace{1em}

\begin{itemize}
\item {\it Monotonicity:} in both cases, the upper bound is independent of $N$ while the lower bound
monotonically approaches to the upper bound as $N$ increases; for $p_{11}\ge p_{01}$, the lower bound converges
to the upper bound as $N\rightarrow\infty$.
\end{itemize}
\label{thm:N3}
\end{theorem}

\begin{proof}
See Appendix~E.
\end{proof}

Numerical results given in~\cite{Liu&Zhao:07TR} have demonstrated the tightness of the bounds:
the relative difference between the lower and the upper bounds is within $6\%$ for a wide range of
transition probabilities~$\{p_{i,j}\}$.

The monotonicity of the difference between the upper and lower bounds with respect to $N$
shows that the performance of the multi-channel
opportunistic system improves with the number $N$ of channels, as suggested by intuition.
For $p_{11}\ge p_{01}$, the upper bound gives the limiting performance of the
opportunistic system when $N\rightarrow\infty$. In Corollary~\ref{thm:rate} below, we show
that the throughput of an opportunistic system increases to a constant at (at least) geometric rate
as $N$ increases.
This result conveys an important message regarding system design:  the throughput of
a multi-channel opportunistic system with single-channel sensing quickly saturates as the number of
channels increases; it is thus crucial to enhance radio sensing capability in order to fully exploit
the communication opportunities offered by a large number of channels.

\begin{corollary}
For $p_{11}> p_{01}$, the lower bound on throughput $U$ converges to the constant upper bound
at geometrical rate $(p_{11}-p_{01})$ as $N$ increases; for $p_{11}< p_{01}$, the lower bound on
$U$ converges to a constant at geometrical rate $(p_{01}-p_{11})^2$.
\label{thm:rate}
\end{corollary}

\begin{proof}
See Appendix~F.
\end{proof}

\section{Conclusion and Future Work}

We have considered an optimal sensing problem that is of fundamental
interest in contexts involving opportunistic communications over
multiple channels. We have
shown that for independent and identically evolving channels, the myopic sensing policy has a simple round-robin
structure, which obviates the need to know the exact channel parameters,
making it extremely easy to implement in practice. We have proved
that the myopic policy is optimal for the two-channel case. We
have also characterized in closed-form the throughput
performance of the myopic policy and the scaling behavior with respect to
the number of channels.

Future directions include sensing policies for non-identical channels
and with multi-channel sensing. In a recent work \cite{Liu&Zhao:08SDR},
the existence of Whittle's index policy and the closed-form expression of
Whittle's index have been obtained, leading to a simple, near-optimal index policy
for non-identical channels with multi-channel sensing.
Furthermore, it is shown in \cite{Liu&Zhao:08SDR} that the myopic policy
is equivalent to Whittle's index policy when channels are identical. The results
obtained in this paper on the myopic policy thus also apply to Whittle's index policy.
The structure and optimality of the myopic policy is also extended to multichannel sensing
in~\cite{Liu&Zhao:08Asilomar}.

It is also of interest to consider sensing policies for
multiple users competing for communication opportunities in multiple
channels. Recent work on extending the myopic sensing policy to multi-user
scenarios can be found in \cite{Liu&etal:08CogNet,Liu&etal:08ISSSTA}.

\section*{Appendix A: Proof of Theorem~\ref{thm:structure}}

We prove Theorem~\ref{thm:structure} by showing that the channel $\hat{a}(t)$ given by
\eqref{eq:structure1} and \eqref{eq:structure2} is indeed the channel with the largest
belief value in slot $t$. Specifically, we prove the following lemma.

\begin{lemma}
Let $\hat{a}(t)=i_1$ be the channel determined by \eqref{eq:structure1} for $p_{11}\ge p_{01}$
and by \eqref{eq:structure2} for $p_{11}< p_{01}$. Let $\Kc(t)=(i_1,i_2,\cdots,i_N)$ be the
circular order of channels in slot $t$, where we set the starting point to $\hat{a}(t)=i_1$.
We then have, for any $t\ge 1$,
\begin{equation}
\omega_{i_1}(t)\ge\omega_{i_2}(t)\ge\cdots\ge\omega_{i_N}(t),
\label{eq:A11}
\end{equation}
\label{lemma:structure}
\ie the channel given by \eqref{eq:structure1} and \eqref{eq:structure2} has the largest
belief value in every slot $t$.
\end{lemma}

To prove Lemma~\ref{lemma:structure}, we introduce operator $\tau(\cdot)$ for the belief update of unobserved channels (see \eqref{eq:omega}).
\begin{equation}
\tau(\omega)\, \defeq\, \omega p_{11}+(1-\omega) p_{01}=p_{01}+\omega (p_{11}-p_{01}).
\label{eq:tau}
\end{equation}
Note that $\tau(\omega)$ is an increasing function of $\omega$ for $p_{11}>p_{01}$
and a decreasing function of $\omega$ for $p_{11}<p_{01}$. Furthermore, we note that
the belief value $\omega_i(t)$ of channel $i$ in slot $t$ is bounded between
$p_{01}$ and $p_{11}$ for any $i$ and $t>1$, and an observed channel achieves either
the upper bound or the lower bound of the belief values (see \eqref{eq:omega}).

We now prove Lemma~\ref{lemma:structure} by induction.
For $t=1$, \eqref{eq:A11} holds by the definition of $\Kc(1)$.
Assume that \eqref{eq:A11} is true for slot~$t$, where
$\Kc(t)=(i_1,i_2,\cdots,i_N)$ and $\hat{a}(t)=i_1$.
We show that it is also true for slot~$t+1$.

Consider first $p_{11}\ge p_{01}$. We have $\Kc(t+1)=\Kc(t)=(i_1,i_2,\cdots,i_N)$.
When $S_{i_1}(t)=1$, we have $\hat{a}(t+1)=\hat{a}(t)=i_1$
from \eqref{eq:structure1}. Since $\omega_{i_1}(t+1)=p_{11}$ achieves the upper bound of the belief values
and the order of the belief values of the unobserved channels remains unchanged due to the monotonically
increasing property of $\tau(\omega)$, we arrive at \eqref{eq:A11} for $t+1$. When $S_{i_1}(t)=0$, we have $\hat{a}(t+1)=i_2$
from \eqref{eq:structure1}. We again have \eqref{eq:A11} by noticing that $\omega_{i_1}(t+1)=p_{01}$ achieves
the lower bound of the belief values and $\Kc(t+1)=(i_2,i_3,\cdots,i_N,i_1)$ when the starting point
is set to $\hat{a}(t+1)=i_2$.

For $p_{11}<p_{01}$, $\Kc(t+1)=-\Kc(t)=(i_1,i_N,i_{N-1},\cdots,i_2)$.
When $S_{i_1}(t)=0$, we have $\hat{a}(t+1)=\hat{a}(t)=i_1$
from \eqref{eq:structure2}. Since $\omega_{i_1}(t+1)=p_{01}$ achieves the upper bound of the belief values
and the order of the belief values of the unobserved channels is reversed due to the monotonically
decreasing property of $\tau(\omega)$, we have, from the induction assumption at~$t$,
\[
\omega_{i_1}(t+1)\ge\omega_{i_N}(t+1)\ge\omega_{i_{N-1}}(t+1)\ge\cdots\ge\omega_{i_2}(t+1),
\]
which agrees with \eqref{eq:A11} for~$t+1$ and $\Kc(t+1)=(i_1,i_N,i_{N-1},\cdots,i_2)$.
When $S_{i_1}(t)=1$, we have $\hat{a}(t+1)=i_N$
from \eqref{eq:structure2}. We again have \eqref{eq:A11} by noticing that $\omega_{i_1}(t+1)=p_{11}$ achieves
the lower bound of the belief values and $\Kc(t+1)=(i_N,i_{N-1},\cdots,i_2,i_1)$ when the starting point
is set to $\hat{a}(t+1)=i_N$. This concludes
the proof of Lemma~\ref{lemma:structure}, hence Theorem~\ref{thm:structure}.

\section*{Appendix B: Last Channel Visits and $j$-Step Transition Probabilities}

As commented in Sec.~\ref{sec:structure}, another way to see the channel switching structure of the myopic policy is
through the last visit to each channel once every channel
has been visited at least once. An alternative proof of this structure
is based on properties of the $j$-step transition probabilities $p_{01}^{(j)}$
and $p_{11}^{(j)}$~\cite{Gallager:95book}.
\begin{eqnarray}
p_{01}^{(j)}&=&\frac{p_{01}-p_{01}(p_{11}-p_{01})^{j}}{p_{01}+p_{10}},\label{eq:kstep1}\\
p_{11}^{(j)}&=&\frac{p_{01}+p_{10}(p_{11}-p_{01})^{j}}{p_{01}+p_{10}}.\label{eq:kstep2}
\end{eqnarray}
It is easy to see that for $p_{11}>p_{01}$,
$p_{01}^{(j)}$ monotonically increases to the stationary
distribution $\omega_o$ as $j$ increases. For
$p_{11}<p_{01}$, $p_{11}^{(j)}$
oscillates around and converges to $\omega_o$ with
$p_{11}^{(j)}>\omega_0$ for even $j$'s and $p_{11}^{(k)}<\omega_0$
for odd $j$'s (see Fig.~\ref{fig:kstep1} and \ref{fig:kstep2}). The channel switching structure thus follows
by noticing that channel switching occurs only after observing $0$ for $p_{11}\ge p_{01}$
and after observing $1$ for $p_{11}<p_{01}$.

\begin{figure}[h]
\centerline{
\begin{psfrags}
\psfrag{a}[r]{ $p_{01}$}
\psfrag{k}[c]{ $j$}
\psfrag{w}[r]{ $\omega_o$}
\psfrag{p}[l]{ $p_{01}^{(j)}$}
\psfrag{L}[c]{ $ $} 
\scalefig{0.4}\epsfbox{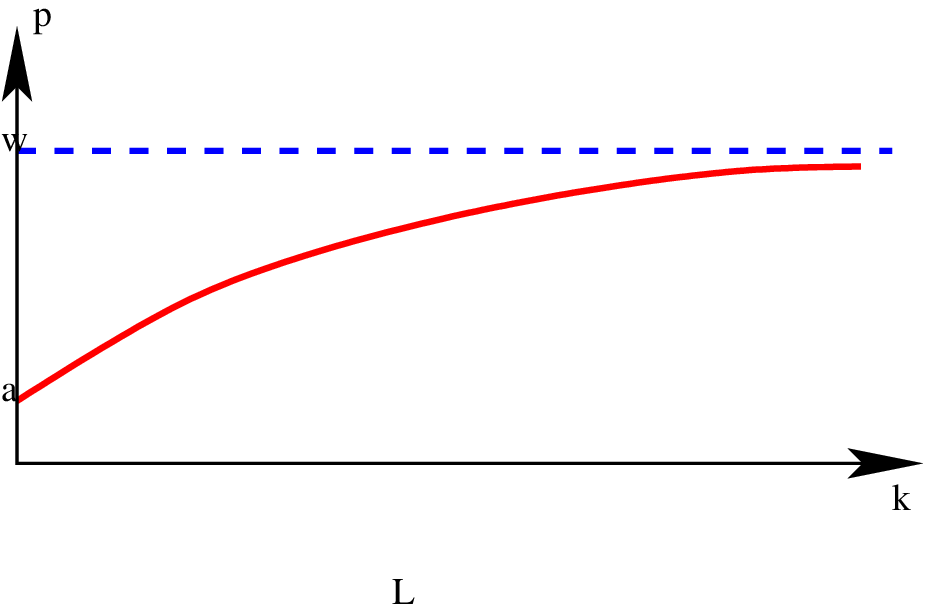}
\end{psfrags}}
\caption{The $j$-step transition probabilities of the Gilbert-Elliot channel when $p_{11}>p_{01}$.}
\label{fig:kstep1}
\end{figure}

\begin{figure}[h]
\centerline{
\begin{psfrags}
\psfrag{a}[r]{ $p_{11}$}
\psfrag{b}[r]{ $p_{01}$}
\psfrag{w}[r]{ $\omega_o$}
\psfrag{k}[c]{ $j$}
\psfrag{c}[c]{ $1$}
\psfrag{d}[c]{ $2$}
\psfrag{e}[c]{ $3$}
\psfrag{f}[c]{ $4$}
\psfrag{p}[l]{ $p_{11}^{(j)}$}
\psfrag{L}[c]{ $ $} 
\scalefig{0.36}\epsfbox{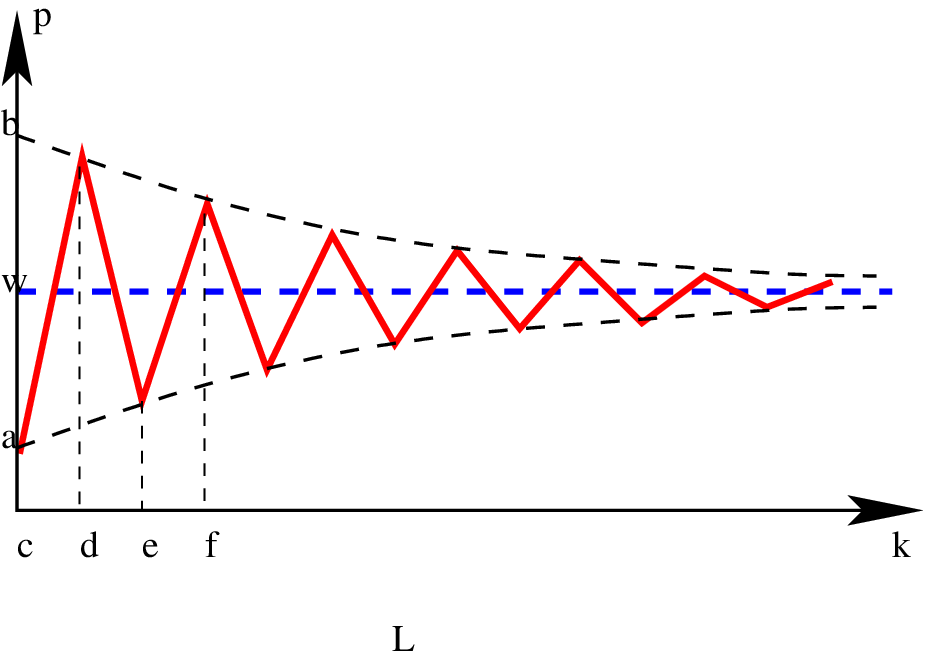}
\end{psfrags}}
\caption{The $j$-step transition probabilities of the Gilbert-Elliot channel when $p_{11}<p_{01}$.}
\label{fig:kstep2}
\end{figure}

\begin{figure*}[t!]
\begin{eqnarray}
\hat{V}_{t}(1|[1,0])&=&p_{01}+p_{10}p_{00}V_{t+1}(2|[0,0])+p_{10}p_{01}V_{t+1}(2|[0,1])+p_{11}p_{00}V_{t+1}(2|[1,0])+p_{11}p_{01}V_{t+1}(2|[1,1]),
\label{eq:AppC1}\\
\hat{V}_{t}(1|[0,1])&=&p_{01}+p_{00}p_{10}V_{t+1}(1|[0,0])+p_{00}p_{11}V_{t+1}(1|[0,1])+p_{01}p_{10}V_{t+1}(1|[1,0])+p_{01}p_{11}V_{t+1}(1|[1,1]).
\label{eq:AppC2}
\end{eqnarray}
\hrulefill
\end{figure*}

\section*{Appendix C: Proof of Theorem~\ref{thm:optimal}}

Recall that $\hat{V}_t(\Omega)$ denotes the total expected reward obtained under the myopic policy
starting from slot $t$. Let $\hat{V}_t(\Omega;a)$ denote the total expected reward obtained by action $a$
in slot $t$ followed by the myopic policy in future slots.
We first establish the following lemma which applies to a general POMDP/MDP.

\begin{lemma}
For a POMDP over a finite horizon $T$, the myopic policy is optimal if for $t=1,\cdots,T$,
\begin{equation}
\hat{V}_t(\Omega)\ge \hat{V}_t(\Omega;a),~~~\forall a,\Omega.
\label{eq:C0}
\end{equation}
\label{lemma:pomdp}
\end{lemma}
\vspace{-1em}
Lemma~\ref{lemma:pomdp} can be proved by backward induction. Specifically, the initial condition $\hat{V}_T(\Omega)=V_T(\Omega)$
is straightforward. Assume that $\hat{V}_{t+1}(\Omega)=V_{t+1}(\Omega)$. We then have, from \eqref{eq:C0},
\begin{eqnarray}
\hat{V}_t(\Omega)&=&\max_{a=1}\{R_a(\Omega)+\sum_{\Omega'} \Pr[\Omega'|\Omega,a]\hat{V}_{t+1}(\Omega)\}\nn\\
&=&\max_{a=1}\{R_a(\Omega)+\sum_{\Omega'} \Pr[\Omega'|\Omega,a] V_{t+1}(\Omega)\}=V_t(\Omega),\nn
\end{eqnarray}
\ie the myopic policy is optimal.

We now prove Theorem~\ref{thm:optimal} based on Corollary~1. Considering all channel state realizations in
slot $t$, we have
\begin{eqnarray}
\hspace{-0.2in}& \hat{V}_t(\Omega;a)=\sum_{\sbf}\Pr[\Sbf(t)=\sbf|\Omega]\hat{V}_t(\Omega;a|\Sbf(t)=\sbf)\nn\\
\hspace{-0.2in}&=\omega_a+\sum_{\sbf}\Pr[\Sbf(t)=\sbf|\Omega]\hat{V}_{t+1}(\Tc(\Omega|a,s_a)|\Sbf(t)=\sbf),
\label{eq:Vta}
\end{eqnarray}
where $\hat{V}_{t+1}(\Tc(\Omega|a,s_a)|\Sbf(t)=\sbf)$ is the {\em conditional} reward obtained
starting from slot $t+1$ given that the system state in slot $t$ is $\sbf$. From Corollary~1, we have
\begin{equation}
\hat{V}_t(\Tc(\Omega|a,s_a)|\Sbf(t-1)=\sbf)=\hat{V}_t(\Tc(\Omega'|a,s_a)|\Sbf(t-1)=\sbf),
\label{eq:prop3a}
\end{equation}
\ie the conditional total expected reward of the myopic policy starting from slot $t$ is determined by the action $a$ in slot $t-1$
and independent of the belief vector $\Omega$ in slot $t-1$ (note that $a(t-1)$ and $\Sbf(t-1)$ determines $\vec{\Sbf}(t)$, which
determines the reward process). Adopting the simplified notation of $\hat{V}_t(a(t-1)|\Sbf(t-1)=\sbf)$, we
further have, from the statistically identical assumption of channels,
\begin{eqnarray}
& & \hat{V}_t(a(t-1)=1|\Sbf(t-1)=[s_1,s_2])\nn\\
&=&\hat{V}_t(a(t-1)=2|\Sbf(t-1)=[s_2,s_1]).
\label{eq:prop3b}
\end{eqnarray}

Next we show that
\begin{eqnarray}
& & \hat{V}_{t}(a(t-1)=1|\Sbf(t-1)=[1,0])\nn\\
&=&\hat{V}_{t}(a(t-1)=1|\Sbf(t-1)=[0,1])).
\label{eq:A1}
\end{eqnarray}
Assume that $p_{01}>p_{11}$.
Following the structure of the myopic policy, we know that the myopic action in slot $t$ is $\hat{a}(t)=2$ for the left hand side
of \eqref{eq:A1} and
$\hat{a}(t)=1$ for the right, which leads to~\eqref{eq:AppC1} and \eqref{eq:AppC2}.
We then have \eqref{eq:A1} based on \eqref{eq:prop3b}. The case of $p_{01}<p_{11}$ can be similarly proved.

Consider $\Omega=[\omega_1,\omega_2]$ with $\omega_1\ge\omega_2$. The myopic action is thus $a=1$. We now establish \eqref{eq:C0}.
From \eqref{eq:Vta} and \eqref{eq:prop3b}, we have
\begin{eqnarray}
\hat{V}_t(\Omega; a=1)&=&\omega_1+\sum_{i,j\in\{0,1\}} \Pr[\Sbf(t)=[i,j]] \hat{V}_{t+1}(1|[i,j]),\nn\\
\hat{V}_t(\Omega; a=2)&=&\omega_2+\sum_{i,j\in\{0,1\}} \Pr[\Sbf(t)=[i,j]] \hat{V}_{t+1}(1|[j,i]).\nn
\end{eqnarray}
It thus follows from \eqref{eq:A1} that
\begin{eqnarray}
& & \hat{V}_t(\Omega; a=1)-\hat{V}_t(\Omega;a=2)\nn\\
&=&(\omega_1-\omega_2)(1+\hat{V}_{t+1}(1|[1,0])-\hat{V}_{t+1}(1|[0,1]))\nn\\
&=&\omega_1-\omega_2\nn\\
&\ge& 0.\nn
\end{eqnarray}
This concludes the proof.

\section*{Appendix D: Proof of Theorem~\ref{thm:N2}}

Consider first $p_{11}\ge p_{01}$. Let $\Rbf=\{r_{i,j}\}$ denote the transition matrix of
$\{L_k\}_{k=1}^{\infty}$, where $r_{i,j}$ is given in~\eqref{eq:rij1}.
Let $\Rbf(:,k)$ denote the $k$-th column of $\Rbf$. We have
\begin{eqnarray}
\textbf{1}-\Rbf(:,1)=\frac{\Rbf(:,2)}{p_{10}},~~~\Rbf(:,k)=\Rbf(:,2)(p_{11})^{k-2},
\label{eq:C1}
\end{eqnarray}
where \textbf{1} is the unit column vector $[1,1,...]^t$. By
the definition of stationary distribution, we have, for $k=1,2,\cdots,$
\begin{equation}
[\lambda_1,\lambda_2,\cdots] \Rbf(:,k)=\lambda_k,
\label{eq:C1b}
\end{equation}
which, combined with \eqref{eq:C1}, leads to
\begin{equation}
\lambda_1=1-\frac{\lambda_2}{(1-p_{11})},~~~\lambda_k=\lambda_2p_{11}^{k-2}.
\label{eq:C2}
\end{equation}
Substituting \eqref{eq:C2} into \eqref{eq:C1b} for $k=2$
and solving for $\lambda_2$, we have $\lambda_2=\bar{\omega}p_{10}$, where $\bar{\omega}$ is given in \eqref{eq:baromega1}.
From \eqref{eq:C2}, we then have
the stationary distribution as
\begin{equation}
\lambda_k= \left\{\begin{array}{ll}
1-\bar{\omega}, & k = 1\\
\bar{\omega}p_{11}^{k-2}p_{10}, & k>1\\
\end{array}
\right.,
\end{equation}
which leads to \eqref{eq:N2} based on \eqref{eq:U} and $\bar{L}=\sum_{k=1}^\infty k\lambda_k$.
The proof for $p_{11}< p_{01}$ is similar based on the transition probabilities given in~\eqref{eq:rij2}.

Based on Corollary~\ref{thm:MRP}, Theorem~\ref{thm:N2} can also be proved by calculating the stationary distribution of $\{\vec{\Sbf}(t)\}$.

\section*{Appendix E: Proof of Theorem~\ref{thm:N3}}

\noindent{\it Case 1: $p_{11}\ge p_{01}$}~~~ Let $\omega_k$ denote the belief value of the chosen
channel in the first slot of the $k$-th TP. The length $L_k(\omega_k)$ of this
TP has the following distribution.
\begin{equation}
\Pr[L_k(\omega_k)=l]= \left\{\begin{array}{ll} 1-\omega_k, & l = 1 \\
\omega_k p_{11}^{l-2}p_{10}, & l>1\\
\end{array}
\right..
\end{equation}
It is easy to see that
if $\omega'\ge \omega$, then
$L_k(\omega')$ stochastically dominates $L_k(\omega)$.

From the round-robin structure of the myopic policy, $\omega_k=p_{01}^{(J_k)}$,
where $J_k=\sum_{i=1}^{N-1} L_{k-i}+1$.
Based on the monotonic increasing property of the $j$-step
transition probability $p_{01}^{(j)}$
(see \eqref{eq:kstep1} and Fig.~\ref{fig:kstep1}), we have $\omega_k\le\omega_o$, where $\omega_o$ is the stationary distribution
of the Gilbert-Elliot channel given in \eqref{eq:wo}. $L_k(\omega_o)$ thus stochastically dominates $L_k(\omega_k)$,
and the expectation of the former, $\overline{L_k(\omega_o)}=1+\frac{\omega_o}{1-p_{11}}$, leads to the upper bound
of $U$ given in~\eqref{eq:N3a}.

Next, we prove the lower bound of $U$ by constructing a hypothetical system where the initial belief value
of the chosen channel in a TP is a lower bound of that in the real system. The average TP length
in this hypothetical system is thus smaller than that in the real system, leading to a lower bound on $U$ based on~\eqref{eq:U}.
Specifically, since $\omega_k=p_{01}^{(J_k)}$ and $J_k=\sum_{i=1}^{N-1} L_{k-i}+1\ge N+L_{k-1}-1$, we have
$\omega_k\le p_{01}^{(N+L_{k-1}-1)}$. We thus construct a hypothetical system given by a first-order Markov chain
$\{L_k'\}_{k=1}^\infty$ with the following transition probability $r_{i,j}$.
\begin{equation}
r_{i,j}=\left\{\begin{array}{ll}
1-p^{(N+i-1)}_{01}, & i\geq 1,~j=1\\
p^{(N+i-1)}_{01}p_{11}^{j-2}p_{10}, & i\geq 1,~j\geq 2
\end{array}
\right..
\end{equation}
It can be shown that the stationary distribution of $\{L_k\}_{k=1}^\infty$
stochastically dominates that of the hypothetical system $\{L_k'\}_{k=1}^\infty$ (see~\cite{Liu&Zhao:07TR} for details).
The latter can be obtained with the same techniques used in Appendix~D. The average
length of $L_k'$ can thus be calculated, leading to the lower bound given in~\eqref{eq:N3a}.

\noindent{\it Case 2: $p_{11}< p_{01}$}~~~ In this case, the larger the initial belief of the chosen channel in a given
TP, the smaller the average length of the TP. On the other hand, \eqref{eq:U} shows that
$U$ is inversely proportional to the average TP length. Thus, similar to the case of $p_{11}\ge p_{01}$,
we will construct hypothetical systems where the initial belief of the chosen channel in a TP is an upper
bound or a lower bound of that in the real system. The former leads to an upper bound on $U$, the latter, a lower bound on $U$.

Consider first the upper bound. From the structure of the myopic policy, it is clear that when $L_{k-1}$ is odd, in the $k$-th
TP, the user will switch to the channel visited in the $(k-2)$-th TP. As a consequence, the initial belief
$\omega_k$ of the $k$-th TP is given by $\omega_k=p_{11}^{(L_{k-1}+1)}$. When $L_{k-1}$ is even, we can show that
$\omega_k\le p_{11}^{(L_{k-1}+4)}$. This is because that for $N\ge 3$ and $L_{k-1}$ even, the user cannot switch to a channel
visited $L_{k-1}+2$ slots ago, and $p_{11}^{(j)}$ decreases with $j$ for even $j$'s and $p_{11}^{(j)}>p_{11}^{(i)}$ for any
even $j$ and odd $i$ (see \eqref{eq:kstep2} and Fig.~\ref{fig:kstep2}). We thus construct a hypothetical system
given by the first-order Markov chain $\{L_k'\}_{k=1}^\infty$ with the following transition probabilities.
\begin{equation}
r_{i,j}=\left\{\begin{array}{ll}
p^{(i+1)}_{11}, & \mbox{ if } i \mbox{ is odd, } j=1\\
p^{(i+1)}_{10}p_{00}^{j-2}p_{01}, & \mbox{ if } i \mbox{ is odd, } j\geq 2\\
p^{(i+4)}_{11}, & \mbox{ if } i \mbox{ is even, } j=1\\
p^{(i+4)}_{10}p_{00}^{j-2}p_{01}, & \mbox{ if } i \mbox{ is even, } j\geq 2\\
\end{array}
\right..\nn
\end{equation}
It can be shown that the stationary distribution of $\{L'_k\}_{k=1}^\infty$
is stochastically dominated by that of $\{L_k\}_{k=1}^\infty$. The former leads to the upper bound of $U$ given in~\eqref{eq:N3b}.

We now consider the lower bound. Similarly, $\omega_k=p_{11}^{(L_{k-1}+1)}$ when $L_{k-1}$ is odd. When $L_{k-1}$ is
even, to find a lower bound on $\omega_k$, we need to find the smallest odd $j$ such that the last visit to the channel
chosen in the $k$-th TP is $j$ slots ago. From the structure of the myopic policy, the smallest feasible odd $j$
is $L_{k-1}+2N-3$, which corresponds to the scenario where all $N$ channels are visited in turn from the $(k-N+1)$-th TP to the
$k$-th TP with $L_{k-N+1}=L_{k-N+2}=\cdots=L_{k-2}=2$. We thus have $\omega_k\ge p_{11}^{(L_{k-1}+2N-3)}$.
We then construct a hypothetical system
given by the first-order Markov chain $\{L_k'\}_{k=1}^\infty$ with the following transition probabilities.
\begin{equation}
r_{i,j}=\left\{\begin{array}{ll}
p^{(i+1)}_{11}, & \mbox{ if } i \mbox{ is odd, } j=1\\
p^{(i+1)}_{10}p_{00}^{j-2}p_{01}, & \mbox{ if } i \mbox{ is odd, } j\geq 2\\
p^{(i+2N-3)}_{11}, & \mbox{ if } i \mbox{ is even, } j=1\\
p^{(i+2N-3)}_{10}p_{00}^{j-2}p_{01}, & \mbox{ if } i \mbox{ is even, } j\geq 2\\
\end{array}
\right..\nn
\end{equation}
The stationary distribution of this hypothetical system leads to the lower bound of $U$ given in~\eqref{eq:N3b}.

\section*{Appendix F: Proof of Corollary~\ref{thm:rate}}

Let $x=|p_{11}-p_{01}|$.
For $p_{11}>p_{01}$,
after some simplifications, the lower bound has the form
$a+b/(x^N+c)$, where $a,b,c~(c\neq0)$ are constants. The upper bound
is $a+b/c$. We have $\frac{|a+b/(x^N+c)-a-b/c|}{x^N}\rightarrow
b/c^2$ as $N\rightarrow\infty$. Thus the lower bound converges to the
upper bound with geometric rate $x$.

For $p_{11}<p_{01}$,
the lower bound has the form $d+e/(x^{2N-1}+f)$, where
$d,e,f~(f\neq0)$ are constants. It converges to $d+e/f$ as
$N\rightarrow\infty$. We have
$\frac{|d+e/(x^{2N-1}+f)-d-e/f|}{x^{2N}}\rightarrow e/(xf^2)$ as
$N\rightarrow\infty$. Thus the lower bound converges with geometric
rate $x^2$.

\section*{Acknowledgement}

The authors would like to thank the associate editor and anonymous reviewers for
their invaluable comments and suggestions.

\bibliographystyle{ieeetr}
{

}

\end{document}